\documentclass[reprint,aps,superscriptaddress,longbibliography,letterpaper,amsmath,final]{revtex4-2}

\usepackage{multirow}%
\usepackage{algcompatible}
\usepackage{algorithmicx}%
\usepackage{algpseudocode}
\usepackage{newfloat}
\usepackage{mathtools}
\DeclareFloatingEnvironment[
  fileext=loa,
  listname=List of Algorithms,
  name=Algorithm,
  placement=tbhp
]{algorithm}

\usepackage{lineno}
\makeatletter
\newcommand*\patchAmsMathEnvironmentForLineno[1]{%
  \expandafter\let\csname old#1\expandafter\endcsname\csname #1\endcsname
  \expandafter\let\csname oldend#1\expandafter\endcsname\csname end#1\endcsname
  \renewenvironment{#1}%
     {\linenomath\csname old#1\endcsname}%
     {\csname oldend#1\endcsname\endlinenomath}}%
\newcommand*\patchBothAmsMathEnvironmentsForLineno[1]{%
  \patchAmsMathEnvironmentForLineno{#1}%
  \patchAmsMathEnvironmentForLineno{#1*}}%
\AtBeginDocument{%
\patchBothAmsMathEnvironmentsForLineno{equation}%
\patchBothAmsMathEnvironmentsForLineno{align}%
\patchBothAmsMathEnvironmentsForLineno{flalign}%
\patchBothAmsMathEnvironmentsForLineno{alignat}%
\patchBothAmsMathEnvironmentsForLineno{gather}%
\patchBothAmsMathEnvironmentsForLineno{multline}%
\makeatletter
\newcommand{\namerefset}[1]{%
  \edef\@currentlabelname{#1}%
}
\makeatother
}
\makeatother

\usepackage{quantikz}
\usepackage{graphicx}
\usepackage{dcolumn}
\usepackage{bm}
\usepackage[colorlinks = true, linkcolor = blue, urlcolor  = blue, citecolor = blue, anchorcolor = blue]{hyperref}
\usepackage[english]{babel}
\usepackage{braket}
\usepackage{amssymb,dsfont}
\usepackage{makecell}
\usepackage{times}

\newtheorem{theorem}{Theorem}

\begin{document}

\title{The Constant Geometric Speed Schedule for Adiabatic State Preparation}

\author{Mancheon Han}
\email{mchan@kias.re.kr}
\affiliation{School of Computational Sciences, Korea Institute for Advanced Study (KIAS), Seoul, 02455, Korea}
\author{Hyowon Park}
\affiliation{Materials Science Division, Argonne National Laboratory, Argonne, IL, 60439, USA}
\affiliation{Department of Physics, University of Illinois at Chicago, Chicago, IL, 60607, USA}
\author{Sangkook Choi}
\email{sangkookchoi@kias.re.kr}
\affiliation{School of Computational Sciences, Korea Institute for Advanced Study (KIAS), Seoul, 02455, Korea}
\date{\today}

\begin{abstract}
  The efficiency of adiabatic quantum evolution is governed by the evolution time $T$, which typically scales as $\mathcal{O}(\Delta^{-2})$ with the minimum energy gap $\Delta$. However, the rigorous lower bound is $\mathcal{O}(L\Delta^{-1})$, where $L$ is the adiabatic path length. Although $L$ is formally upper-bounded by $\mathcal{O}(\Delta^{-1})$, such a bound is often too loose in practice, and $L$ can be bounded independently of $\Delta$. This indicates the potential for a quadratic speedup through adiabatic schedule construction. Here, we introduce the constant geometric speed (CGS) schedule, which traverses the adiabatic path at a uniform rate. We show that this approach reduces the scaling of the evolution time by a factor of $\Delta^{-1}$, provided $L$ remains bounded independently of $\Delta$. We propose a segmented CGS protocol where path segment lengths are computed from eigenstate overlaps on the fly, reducing the prior spectral-knowledge requirement from the full gap function $\Delta(s)$ to just a global lower bound on the energy gap. Numerical tests on adiabatic unstructured search, N$_2$, and a [2Fe-2S] cluster demonstrate the optimal $\Delta^{-1}$ scaling, confirming a quadratic speedup over the standard linear schedule.
\end{abstract}

\maketitle

\section{Introduction}
The adiabatic theorem~\cite{BornFock1928, Kato1950, Jansen2007} states that
a quantum system remains in its instantaneous eigenstate as long as the state
of interest is separated by an energy gap from other eigenstates and the Hamiltonian changes
sufficiently slowly. This principle underlies both analog
adiabatic quantum computation~\cite{Farhi2000,Georgescu2014,Albash2018} and
its digital counterpart, adiabatic state preparation (ASP)~\cite{Aspuru-Guzik2005,Veis2014,Keever2024}, which implements
adiabatic evolution through quantum circuits~\cite{NielsenAndChuang2010}
and is widely used as a state-preparation subroutine for algorithms such as
quantum phase estimation~\cite{Kitaev1995,Abrams1999}. In both cases, the key
resource is the \emph{adiabatic evolution time} $T$, defined here as the minimum duration
required for the evolved state $\ket{\psi(T)}$ to achieve a target fidelity
$\mathcal F = |\braket{\psi(T)|\Phi_f}|^2$ with the desired eigenstate
$\ket{\Phi_f}$.

Among various parameters that affect the required adiabatic evolution time \(T\),
the minimal energy gap \(\Delta\) plays a major role. Both lower and upper bounds for \(T\) in terms of \(\Delta\) are established~\cite{Albash2018,ChenPRR2023}. The adiabatic evolution
starts with the eigenstate \(\ket{\Phi_i}\) of the easily solvable
Hamiltonian \(H_i\), and aims to obtain the eigenstate \(\ket{\Phi_f}\) of \(H_f\),
the target Hamiltonian. The rigorous adiabatic theorem~\cite{Jansen2007}
applies to any variation of the Hamiltonian between \(H_i\) and \(H_f\)
as long as \(H(t)\), the Hamiltonian as a function of time \(t\),
is a function of the renormalized time \(\tau = t/T\)
that starts with \(H_i\) and ends at \(H_f\).
Here we narrow down our discussion to the following form ,
\begin{align}
  H(s(\tau)) = H_i  + (H_f - H_i) s(\tau), \label{eq:Hs}
\end{align}
The parameter \(s\) is a monotonically
increasing function of the normalized time \(\tau\), referred to as the schedule. The schedule \(s(\tau)\) maps the normalized time \(\tau \in [0,1]\) to the interpolation parameter \(s \in [0,1]\), satisfying \(s(0) = 0\) and \(s(1) = 1\).

For the given adiabatic path \(\{\ket{\Phi(s)}\}_{s\in[0,1]}\) of instantaneous eigenstates \(\ket{\Phi(s)}\) of \(H(s)\), the lower bound for \(T\) is given by ~\cite{Boixo2010}
\begin{align}
  T > \mathcal{O}(L/\Delta)
  \label{eq:T_lowerbound}
\end{align}
where \( L = \int_0^1 \left\| \partial_s\ket{\Phi(s)} \right\| ds \) is
the adiabatic path length. As \(L\) is often can be bounded independently
of \(\Delta\)~\cite{Boixo2010}, the required adiabatic evolution time \(T\)
scales at least as \(\mathcal{O}(\Delta^{-1})\).
On the other hand, a sufficient condition for adiabaticity~\cite{Jansen2007,Albash2018}
yields an upper bound
\begin{align}
T < O\!\left(\frac{\|\dot H\|}{\Delta^2}\right)
+ \mathcal O\!\left(\frac{\|\ddot H\|}{\Delta^2}\right)
+ \mathcal O\!\left(\frac{\|\dot H\|^2}{\Delta^3}\right),
\label{eq:T_wo_schedule}
\end{align}
where \(\dot{\phantom{v}}\) denotes a derivative with respect to \(\tau\). This condition implies that \(T\) scales at most as
\(\mathcal{O}(\Delta^{-2})\) or even \(\mathcal{O}(\Delta^{-3})\).
These different scaling behaviors between upper and lower bound highlights the potential role of scheduling, i.e., the choice of \(s(\tau)\).
For example, in the adiabatic Grover search, a linear schedule gives
$T=\mathcal{O}(\Delta^{-2})$, whereas an optimized schedule using a priori
gap information achieves the optimal $\mathcal{O}(\Delta^{-1})$ scaling~\cite{Roland2002}. 
Since then, a number of optimal scheduling strategies have been 
proposed~\cite{Rezakhani2009,Rezakhani2010,Isermann2021,Matsuura2021,Shingu2025,Jarret2019,Braida2025,Boixo2009,ChenPRR2022}.
However, they require spectral information or lack general performance guarantees,
which limits their practical applicability.

In this work we introduce a geometric framework for constructing schedules that
improve the gap scaling.  Motivated by the geometric origin of the optimal
$\mathcal{O}(\Delta^{-1})$ scaling, we propose the \emph{constant-speed
schedule}, which traverses the adiabatic path at a uniform rate and improves
the scaling of the upper bound of \(T\) in \(1/\Delta\) by one order,
provided that \(L\) remains bounded independently of \(\Delta\).
We then propose a practical classical-quantum hybrid algorithm to
implement this schedule using eigenstate overlaps computed along the evolution.
We demonstrate the method on the adiabatic
Grover algorithm, N$_2$ molecule, and a [2Fe-2S] cluster,
showing that the constant geometric speed schedule achieves the optimal
$\mathcal{O}(\Delta^{-1})$ scaling, a quadratic speedup over
$\mathcal{O}(\Delta^{-2})$ scaling of the linear schedule.
In contrast to earlier approaches, our method reduces the prior
spectral-knowledge requirement from the full gap function $\Delta(s)$ to
just a global lower bound on the energy gap, and provably improves the gap
dependence of the adiabatic evolution.

\section{Notation}
Before presenting the methods, we briefly summarize the notation.  
The Hamiltonian along the adiabatic path is denoted by $H(s)$, with 
$H_i$ and $H_f$ the endpoints and $\ket{\Phi_i}$, $\ket{\Phi_f}$ their 
respective eigenstates.  
The state at physical time $t$, $\ket{\psi(t)}$, evolves from the initial 
state $\ket{\Phi_i}$ according to the Schrödinger equation with $H(s(\tau))$.  
When no confusion arises, we write $H(\tau)\equiv H(s(\tau))$ and 
$\ket{\Phi(\tau)}$ for the instantaneous eigenstate $\ket{\Phi(s(\tau))}$.

\section{Adiabatic evolution error}
With $\mathcal F = |\!\braket{\psi(T)|\Phi_f}\!|^2$,
the adiabatic evolution error satisfies~\cite{Jansen2007}
\begin{align}
1-\mathcal F \le \left( \frac{\mathcal C[s]}{T} \right)^2 ,
\end{align}
where the functional $\mathcal C[s]$ depends on the schedule $s(\tau)$ and is
given by
\begin{align}
\mathcal C[s] &= \sum_{\tau\in\{0,1\}} \frac{\|\dot P\|}{\Delta(\tau)}
 + \int_0^1 \frac{\|Q\ddot P P\|}{\Delta(\tau)}\, d\tau \nonumber \\
&\quad + \int_0^1 \frac{\|\dot P\|^2}{\Delta(\tau)}\, d\tau
 + \int_0^1 2 \frac{\|\dot H\|\, \|\dot P\|}{\Delta^2(\tau)}\, d\tau ,
\label{eq:infidelity_bound}
\end{align}
with $P(\tau)=\ket{\Phi(\tau)}\!\bra{\Phi(\tau)}$ and $Q(\tau)=\mathbb I-P(\tau)$.
For clarity, the explicit $\tau$ dependence has been omitted. The evolution time $T$ required to reach a desired fidelity
is therefore proportional to $\mathcal C[s]$. Without exploiting the geometric structure of the eigenstate path, the terms in
Eq.~\eqref{eq:infidelity_bound} can be bounded using derivatives of the
Hamiltonian and the energy gap~\cite{Jansen2007} as
$\|\dot P\|\le \|\dot H\|/\Delta$ and
$\|Q\ddot P P\|\le \|\ddot H\|/\Delta + 4\|\dot H\|^2/\Delta^2$.
Substituting these estimates yields upper bounds in Eq.~(\ref{eq:T_wo_schedule}).

\section{Constant geometric speed schedule}
The lower bound of the adiabatic evolution time is determined by geometric
properties of the eigenstate path~\cite{Boixo2010}, suggesting that the
scaling with respect to the minimum gap can be improved by rewriting
Eq.~\eqref{eq:infidelity_bound} in geometric form.  The natural 
parameterization of the path is its Fubini--Study arc length~\cite{ONeill2006} 
$l(s)=\int_0^s \|\partial_{s'}\Phi(s')\| ds'$. Accordingly, the instantaneous speed is
$v(\tau)=dl/d\tau=\|\dot\Phi(\tau)\|$.  Since $\|\dot P\|=v(\tau)$, the first
term of $\mathcal C[s]$ directly reflects this speed.
The second term is associated with the
curvature $\kappa$ of the adiabatic path, which can be defined as~\cite{Alsing2024}
$\kappa(\tau)=\|Q\,d^2\Phi/dl^2\|$.  Using the curvature, we have 
$\|Q\ddot P P\|=\sqrt{\dot v^2+\kappa^2 v^4}$
(Proof in Sec.~\uppercase{i}\,A of the Supplemental Material (SM)~\cite{SM}).
Thus, $\mathcal{C}[s]$ in Eq.~\eqref{eq:infidelity_bound} can be written with
geometric terms as
\begin{align}
\mathcal{C}[s] &= \sum_{\tau \in \{0,1\}} \frac{v(\tau)}{\Delta(\tau)}
 + \int_0^1 \frac{\sqrt{\dot{v}^2(\tau) + \kappa^2(\tau) v^4(\tau)}}{\Delta(\tau)} d\tau
   \nonumber \\
&\quad + \int_0^1 \frac{v^2(\tau)}{\Delta(\tau)} d\tau
 + \int_0^1 2 \frac{\|\dot{H}(\tau)\| v(\tau)}{\Delta^2(\tau)} d\tau.
   \label{eq:infidelity_bound_with_speed}
\end{align}
For the constant geometric speed (CGS) schedule $s_c(\tau)$, the speed becomes uniform, $v(\tau)=L$ with $L=\int_0^1 v(\tau)d\tau$, and the acceleration term
$\dot v$ vanishes.  Substituting into
Eq.~\eqref{eq:infidelity_bound_with_speed} yields
\begin{align}
\mathcal{C}[s_c] =&L\sum_{\tau\in\{0,1\}}\frac{1}{\Delta(\tau)}+ L^2 \int_0^1
\frac{\kappa(\tau) + 1}{\Delta(\tau)} d\tau \nonumber \\
&+ 2L\int_0^1 \frac{\| \dot{H}(\tau) \|}{\Delta^2(\tau)} d\tau.
\end{align}
This leads to the following scaling for the total time $T$:
\begin{align}
    T = \mathcal{O}\!\left( \frac{L}{\Delta} \right)
      + \mathcal{O}\!\left(\frac{L(K+L)}{\Delta}\right)
      + \mathcal{O}\!\left(\frac{L \|\dot{H}\|}{\Delta^2}\right),
       \label{eq:T_css}
\end{align}

where \(K = \int_0^L \kappa(l) \, dl\) is the total curvature of the adiabatic path.
Comparing this with Eq.~\eqref{eq:T_wo_schedule}, every term in the bound improves by one order in $\Delta$.
Here, the explicit dependence on the gap is replaced by the geometric quantities $L$ and $K$.
If these geometric quantities exhibit worst-case scaling (e.g., $L \propto \Delta^{-1}$), the advantage of the CGS schedule vanishes.
However, in many cases of interest, $L$~\cite{Boixo2010} and $K$ can be bounded independently of $\Delta$ (see Figs.~\ref{fig:3}(c) and \ref{fig:4}(d)).
In this regime, the CGS schedule achieves a one-order improvement in gap scaling.
Consequently, the constant-speed schedule can yield the optimal $T=\mathcal{O}(\Delta^{-1})$ scaling in cases where generic schedules typically require $T=\mathcal{O}(\Delta^{-2})$, such as for Hamiltonians belonging to the Gevrey class~\cite{Elgart2012}.

The constant geometric speed condition ($\dot{v}=0$) implies $ds/d\tau \propto (dl/ds)^{-1}$.
Since the geometric speed is bounded by $dl/ds \leq \|H_f-H_i\|/\Delta$, this suggests a similarity to previous gap-based locally optimal schedules $ds/d\tau \propto \Delta^p$~\cite{Jansen2007, Roland2002, Albash2018}.
This similarity is supported by our numerical experiments (Fig.~S2 in the SM~\cite{SM}).
However, despite this similarity in performance, our method offers a distinct practical advantage.
Unlike the gap function $\Delta(s)$, which requires knowledge of excited states, the geometric speed is determined solely by the ground state.
Consequently, the CGS schedule can be computed on the fly during the adiabatic evolution, as the evolving state provides a good approximation of the true ground state.
The only parameter required is a global lower bound of the energy gap $\Delta$ to filter out excited state components.
This parameter is far more tractable than characterizing $\Delta(s)$ for all $s$, and can even be estimated adaptively.
A related geometric formulation by Chen~\cite{ChenPRR2022} likewise casts the
optimal protocol as constant-velocity traversal, but with respect to a
\emph{dynamical} quantum geometric tensor, in contrast to the standard
Fubini--Study metric used in our work. This choice of metric allows our
method to compute the adiabatic path length and the corresponding schedule
directly through eigenstate overlaps.


We now describe how to construct such a schedule in practice
based on the length of adiabatic path segments,
\begin{align}
\Delta l = \int_s^{s + \Delta s} \left\| \frac{d\ket{\Phi(s')}}{ds'} \right\| ds'.
\label{eq:delta_l}
\end{align}
The equation suggests that we can traverse the adiabatic path at a uniform rate 
by keeping the ratio $\Delta l / \Delta t$ constant.
The important thing to remark is that the segment length $\Delta l$ can be
obtained from the overlap between nearby
eigenstates, from $|\!\braket{\Phi(s) | \Phi(s + \Delta s)}\!|^2 = 1 - (\Delta l)^2$
up to the leading order.
Then, the following theorem provides a discretized approximation
to the constant geometric speed schedule (Proof in Sec.~I\,B of SM).
\begin{theorem}
\label{theorem:css}
Suppose a monotonic sequence \( s_0 = 0 < s_1 < \dots < s_m = 1 \) is given, and that the eigenstate overlaps 
\[
|\braket{\Phi(s_j) | \Phi(s_{j+1})}|^2 = 1 - (\Delta l_{j+1})^2
\]
are known for each adjacent pair \( (s_j, s_{j+1}) \). Then, any schedule
\( \hat{s}(\tau) \) satisfying \( \hat{s}(\tau_j) = s_j \) and  
\[
\tau_{j+1} = \tau_j + \frac{\Delta l_{j+1}}{\sum_{i=1}^m \Delta l_i}
\]
converges to the constant geometric speed schedule \( s_c \) in the limit \( \max_j (s_{j+1} - s_j) \rightarrow 0 \).
\end{theorem}

Based on Theorem~\ref{theorem:css}, we can construct a segmented constant geometric speed schedule, $\hat{s}_c(\tau)$, by interpolating the discrete points $\{(\tau_j, s_j)\}$.
This schedule serves as a practical, discretized realization of a true constant
speed schedule. The core challenge in this approach lies in the computation
of the path segment lengths $\Delta l_j$, which requires calculating the eigenstate overlaps
$|\braket{\Phi(s_j) | \Phi(s_{j+1})}|^2$. The practical algorithm for this task is
given in the Appendix.

\begin{figure} 
\centering
\centerline{\includegraphics[width=8.6cm]{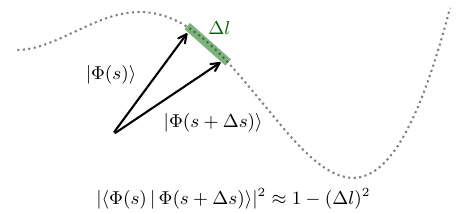}}
\caption{Schematic construction of the constant geometric speed (CGS) schedule. 
The segment length $\Delta l$ is obtained from the overlap 
$|\langle\Phi(s)|\Phi(s+\Delta s)\rangle|^2$, then $\Delta t$ is adjusted so that $\Delta l/\Delta t$ remains constant.}
\label{fig:1}
\end{figure}

\begin{figure*} 
\centering
\centerline{\includegraphics[width=17.2cm]{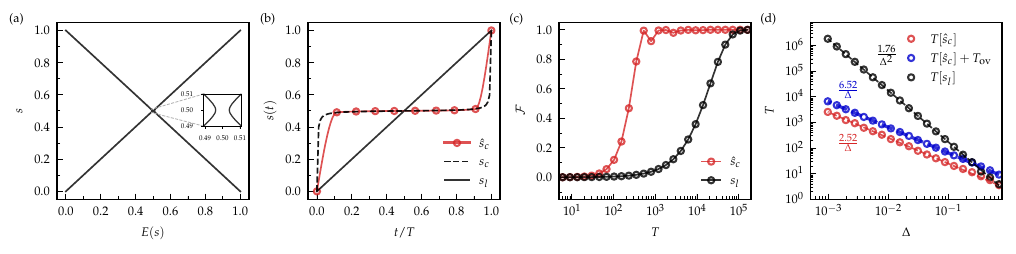}}
\caption{Application of the constant-speed schedule to the adiabatic Grover search.
Panels (a--c) show results for $N=2^{14}$.
(a) Energy spectrum of $H(s)$.
(b) Segmented constant-speed schedule $\hat{s}_c$ compared with the exact optimal schedule $s_c$ and the linear schedule $s_l$; circles denote computed points.
(c) Fidelity $\mathcal{F}$ versus evolution time $T$ for $\hat{s}_c$ and $s_l$.
(d) Evolution time required to reach $\mathcal{F}=0.75$ as a function of the minimum gap $\Delta$; circles show numerical data and dashed lines indicate fitted scalings.
}
\label{fig:2}
\end{figure*}

\section{Applications}
We apply our method to a variety of quantum systems, ranging from quantum search
problems to quantum chemistry. In the following analysis, we account for the overlap calculation
overhead, $T_{\mathrm{ov}}$, which we estimated as $4\Delta^{-1}$
(details of this estimate and other computational parameters
are provided in the Appendix).
For chemical systems, the B3LYP exchange-correlation functional~\cite{Becke1993,Lee1988}
is used in a density functional theory (DFT) calculation.

\subsection{Adiabatic Grover search}
The Grover search problem~\cite{Grover1997} is a canonical example for adiabatic scheduling.
Its adiabatic formulation~\cite{Roland2002} evolves the initial state $\ket{\Phi_i}=\ket{+}^{\otimes n}$ toward the marked computational basis state $\ket{m}$ using the Hamiltonians $H_i=\mathbb I-\ket{\Phi_i}\!\bra{\Phi_i}$ and $H_f=\mathbb I-\ket{m}\!\bra{m}$, interpolated as in Eq.~\eqref{eq:Hs}.
While a linear schedule $s(\tau)=\tau$ yields $T=\mathcal{O}(N)$, offering no advantage over classical search, the optimal schedule achieves a quadratic speedup.
However, realizing this optimal schedule traditionally relies on the condition $\dot{s}\propto\Delta^{2}(s)$~\cite{Roland2002}, which necessitates explicit prior knowledge of $\Delta(s)$.
Crucially, we prove that this optimal schedule is, in fact, identical to the CGS schedule (Proof in Sec.~I\,C of SM).
Consequently, our method naturally reproduces the optimal Grover schedule by computing local geometric information on-the-fly, without requiring \textit{a priori} knowledge of $\Delta(s)$.
%

Figure~\ref{fig:2} (a-c) shows the numerical results for the adiabatic Grover problem with $N=2^{14}=16384$.
Panel (a) displays the energy spectrum, which exhibits its minimum gap at $s=0.5$.
Consistent with this feature, our algorithm generates a schedule $\hat{s}_c(\tau)$ that slows down near the gap minimum, as shown in panel (b).
The computed schedule points (red) lie almost exactly on the theoretical constant-speed schedule $s_c(\tau)$, with deviations arising only from interpolation.
Panel (c) quantifies the resulting speedup: our schedule reaches $75\%$
fidelity approximately $92$ times faster than the linear schedule ($T=3.19\times10^2$ versus $T=2.94\times10^4$).

Finally, Fig.~\ref{fig:2}(d) confirms the asymptotic scaling of the evolution
time with the minimum gap on the adiabatic path \(\Delta\).
For the adiabatic Grover search, \(\Delta=1/\sqrt{N}\)~\cite{Roland2002},
so the plot also directly connects
evolution time with the number of items \(N\) in the database .
We plot the adiabatic evolution time $T$
and the total time $T_\text{tot} = T + T_\mathrm{ov}$. The numerical fit shows that the linear
schedule's evolution time scales as $T \propto \Delta^{-2} = \mathcal{O}(N)$,
while our constant geometric speed schedule achieves the optimal scaling of $T \propto
\Delta^{-1} = \mathcal{O}(\sqrt{N})$. Since the overlap calculation overhead $T_\mathrm{ov}$ also
scales as $\Delta^{-1}$, the total time $T_\text{tot}$ maintains this optimal
scaling, with a larger prefactor.

\subsection{N$_2$}
The first system we consider is the nitrogen molecule, described using a STO-3G basis set~\cite{Stewart1970} and restricted to the
singlet subspace. As the initial Hamiltonian $H_i$, we use the Kohn-Sham Hamiltonian from a converged DFT calculation. The system then
evolves towards the final Hamiltonian $H_f$, which is the full electronic Hamiltonian within this basis set. Figure~\ref{fig:3}(a) shows the
two lowest-lying energies as a function of bond length $R$, obtained from the DFT calculation corresponds to $H_i$ and from a Full Configuration
Interaction (FCI) calculation of $H_f$. The inset reveals that the energy gap of the initial Hamiltonian decreases rapidly with increasing $R$,
causing the overall minimum gap $\Delta$ of the adiabatic path to occur near $s=0$ for large $R$.

Our CGS schedule $\hat{s}_c$ allocates a significantly larger fraction of the evolution time to this region, resulting in a substantial reduction in the required evolution time, as shown for $R=3.5\,\text{\AA}$ in Fig.~\ref{fig:3}(b) (see SM~\cite{SM} Fig.~S1(a) for the shape of $\hat{s}_c$).
To be quantitative, let $T$ denote the time required to reach a target fidelity of $0.75$.
We find that our schedule yields a speedup of $T[s_l] / T[\hat{s}_c] \approx 52.6$ over the linear schedule $s_l$.
To investigate the origin of this speedup, in Fig.~\ref{fig:3}(c), we plot the adiabatic path length $L$ and total curvature $K$ as a function of the minimum gap $\Delta$, obtained by varying the bond length $R$.
We observe that both $L$ and $K$ remain bounded independently of $\Delta$, unlike the theoretical worst-case scaling.
This behavior accounts for the performance of CGS and leads to the reduced gap scaling from $\Delta^{-2}$ to $\Delta^{-1}$, as confirmed in Fig.~\ref{fig:3}(d).

\begin{figure*} 
\centering
\centerline{\includegraphics[width=17.2cm]{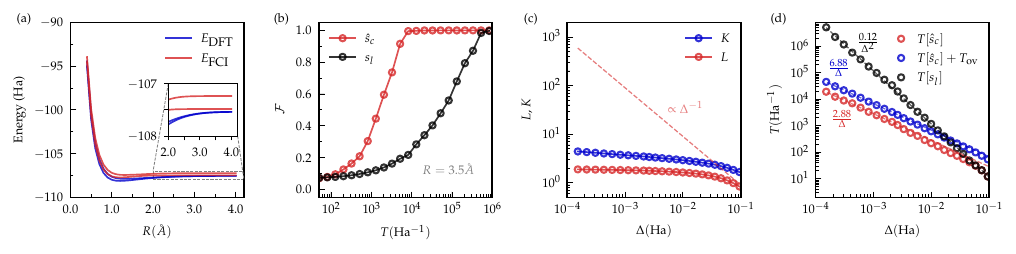}}
\caption{Application of the CGS schedule to the nitrogen molecule.  
(a) Low-lying energy spectra from Density Functional Theory (DFT, blue) and Full Configuration Interaction (FCI, red) calculations.  
(b) Fidelity \(\mathcal{F}\) as a function of the evolution time \(T\) at bond length \(R = 3.5\,\text{\AA}\).  
(c) The adiabatic path length \(L\) and the total curvature \(K\)
as a function of $\Delta$. Dotted red line corresponds to worst-case scaling of \(L\).
(d) Evolution time needed to reach $\mathcal{F}=0.75$ versus the minimum gap $\Delta$.
The constant-speed schedule (red) achieves the optimal $\Delta^{-1}$ scaling, while the linear schedule (black) follows $\Delta^{-2}$.
The total time including overlap-overhead (blue) also follows the optimal scaling.
}
\label{fig:3}
\end{figure*}

\subsection{[2Fe-2S]}
We demonstrate our method on a [2Fe-2S] cluster, a common bioinorganic
motif and a representative strongly correlated system in
quantum chemistry~\cite{Beinert1997, Johnson2005, Reiher2017, Sharma2014}.
Following Lee et al.~\cite{Lee2023}, we test a wide range of initial
Slater determinants and reproduce their observation that ASP performance under
a linear schedule is highly unpredictable: the required evolution time
varies over eight orders of magnitude depending on the initial state and
often exceeds the estimated cost of quantum phase estimation ($T_{\textrm{QPE}}$)~\cite{Lee2023}.
This extreme sensitivity arises from the $T\propto\Delta^{-2}$ scaling,
which strongly amplifies variations in the minimum gap $\Delta$.

Our constant-speed schedule mitigates this issue by achieving the optimal $T\propto\Delta^{-1}$ scaling. To quantitatively compare the two approaches across many initial states, we use the heuristic estimate~\cite{Lee2023}
\begin{align}
  T^{\textrm{est}}_\textrm{ASP} = \max_{\tau\in[0,1]}\frac{|\braket{\Phi(\tau)|\dot{H}(\tau)|\mathcal{E}(\tau)}|}{\Delta^2(\tau)},
\end{align}
where $\ket{\mathcal{E}(\tau)}$ is the first excited state at $\tau$.

The results in Fig.~\ref{fig:4}(b) confirm that the optimal gap scaling leads to improved reliability. 
While $T^{\textrm{est}}_\textrm{ASP}$ for the linear schedule (black) shows the previously reported large variation, the CGS schedule (red) exhibits much smaller variance
and consistently stays below $T_\textrm{QPE}$ (dotted line). 
By achieving optimal scaling, our method thus makes ASP robust and less unpredictable. 

Figure~\ref{fig:4}(c) presents the actual adiabatic evolution results for a representative case, where the initial state corresponds to the DFT ground state [indicated by crosses in Fig.~\ref{fig:4}(b)].
The estimates in (b) suggest a speedup of more than two orders of magnitude with the CGS schedule, and the simulation in (c) confirms this expectation (see SM~\cite{SM} Fig.~S1(b) for the shape of $\hat{s}_c$).
Using the CGS schedule, the system achieves a fidelity of 75\% with $T=9.48\times 10^3$, which is approximately 289 times faster than the linear schedule ($T=2.74\times 10^6$).
Consistent with the findings in the N$_2$ molecule, the observed performance gain and the reduced gap scaling are directly attributed to the geometric properties shown in Fig.~\ref{fig:4}(d), where both the adiabatic path length $L$ and the total curvature $K$ remain bounded independently of the minimum gap $\Delta$.

\begin{figure*} 
\centering
\centerline{\includegraphics[width=17.2cm]{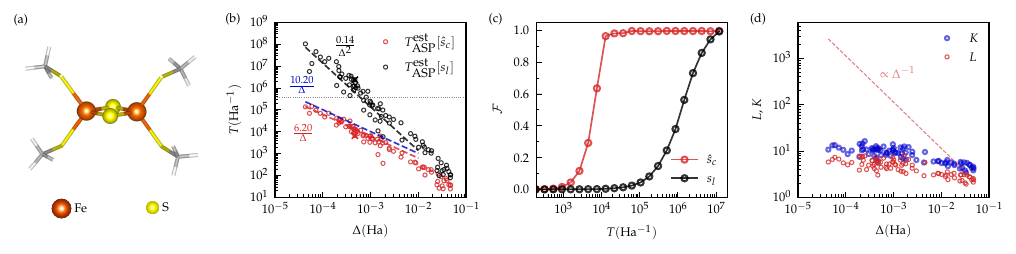}}
\caption{Application of the CGS schedule to the [2Fe-2S] cluster.  
(a) Molecular structure.  
(b) Estimated adiabatic evolution time \(T_\textrm{ASP}^\textrm{est}\) versus the minimum energy gap \(\Delta\), showing the optimal \(\Delta^{-1}\) scaling for \(\hat{s}_c\) (red) compared with the \(\Delta^{-2}\) scaling of \(s_l\) (black).  
(c) Fidelity \(\mathcal{F}\) as a function of evolution time \(T\)
for the initial state chosen as the DFT ground state, corresponding to the crosses in (b).  
Plot shows a speedup of more than two orders of magnitude with \(\hat{s}_c\).
(d) The adiabatic path length \(L\) and the total curvature \(K\)
as a function of $\Delta$. Dotted red line corresponds to worst-case scaling of \(L\).}
\label{fig:4}
\end{figure*}

\section{Conclusion}
We have introduced the constant geometric speed (CGS) schedule as a general strategy to improve the gap scaling of adiabatic state preparation. By enforcing uniform evolution along the eigenstate path, the method reduces the gap dependence of the runtime by one order compared with generic schedules, provided that the adiabatic path length $L$ remains bounded independently of $\Delta$. Furthermore, its overlap-guided construction makes the approach practical without requiring \textit{a priori} spectral knowledge.
Our numerical results on the adiabatic Grover search, the nitrogen molecule, and the [2Fe-2S] cluster demonstrate a clear quadratic speedup, improving the adiabatic scaling from $\mathcal{O}(\Delta^{-2})$ to $\mathcal{O}(\Delta^{-1})$ and substantially enhancing the reliability of state preparation.
Our numerical experiments indicate that the gap-independent boundedness of geometric quantities $L$ and $K$ is a typical feature in practice, supporting the broad applicability of the proposed method.
While this empirical observation emerges consistently across our tested systems, a rigorous characterization of the Hamiltonian classes exhibiting such gap-independent geometric boundedness, together with a formal proof of schedule optimality, is left for future work. Overall, our findings indicate that exploiting the geometry of the adiabatic path provides an effective strategy for improving quantum simulation.

\begin{acknowledgments}
We used resources of
the Center for Advanced Computation at Korea
Institute for Advanced Study
and the National Energy Research Scientific Computing Center (NERSC), a U.S. Department of Energy
Office of Science User Facility operated under
Contract No. DE-AC02-05CH11231.
SC was supported by a KIAS Individual Grant (No. CG090601)
at Korea Institute for Advanced Study and by Institute of Information \& Communications
Technology Planning \& Evaluation (IITP) grant
funded by the Korea government (MSIT) (No.~2022-0-01026).
M.H. is supported by a KIAS Individual Grant (No. CG091302) at Korea Institute for Advanced Study and by Quantum Simulator Development Project for Materials Innovation through
the National Research Foundation of Korea (NRF) funded
by the Korean government (Ministry of Science and ICT(MSIT))(No. NRF-2023M3K5A1094813).
\end{acknowledgments}

\textit{Data and code availability.}---The numerical data, figures, and the code used to reproduce them are openly available at Ref.~\cite{CodeRepo}.

\appendix
\section{Practical algorithm and complexity analysis}
Here, we discuss the practical quantum algorithm to construct the segmented CGS
$\hat{s}_c(\tau)$ based on Theorem~\ref{theorem:css}.
The key step is to compute the eigenstate overlaps $|\braket{\Phi(s_j)|\Phi(s_{j+1})}|^2$.
One approach to calculate this overlap is through the use of projection operators,
$P(s) = \ket{\Phi(s)}\bra{\Phi(s)}$. The squared overlap can be computed from the
ratio of two expectation values:
\begin{align}
|\braket{\Phi(s_j)|\Phi(s_{j+1})}|^2 = \frac{p_j(s_{j+1})}{f_j},
\label{eq:overlap_ratio}
\end{align}
where
\begin{align}
f_j &= \|P(s_j)\ket{\psi(t_j)}\|^2, \label{eq:fj_def} \\
p_j(s_{j+1}) &= \|P(s_{j+1}) P(s_j)\ket{\psi(t_j)}\|^2.
\label{eq:pj_def}
\end{align}
The quantity $f_j$ serves as a fidelity estimate between the evolved state
$\ket{\psi(t_j)}$ and the instantaneous eigenstate $\ket{\Phi(s_j)}$, and it remains
non-negligible as long as the evolution time is sufficiently long. Consequently, the
segment length is computed as
\begin{align}
\Delta l_{j+1} = \sqrt{1 - \frac{p_j(s_{j+1})}{f_j}}.
\label{eq:segment_length_from_proj}
\end{align}
This procedure provides a general framework for constructing the schedule, assuming
one has a method to implement the required projection operators. The entire procedure
is summarized in Algorithm~\ref{alg:css_qzmc}.

At first glance, it might appear simpler to compute segment lengths at pre-determined uniform grid points rather than using the root-finding procedure in Algorithm~\ref{alg:css_qzmc}. However, as seen in the adiabatic Grover problem (Fig.~2),
the energy gap \(\Delta(s)\) often becomes small only in a very narrow region of \(s\),
so a uniform grid may fail to capture the point where long evolution time is required.
In contrast, our root-finding procedure adaptively places grid points in proportion to the required time, ensuring that such regions are properly resolved.

A concrete implementation of the projection operators required in
Algorithm~\ref{alg:css_qzmc} can be realized using the Quantum Zeno
Monte Carlo (QZMC) method~\cite{mchan2025QZMC}. The QZMC method
approximates the projector with a Gaussian that can be represented as a integral of
time evolution,
\begin{align}
P(s) &\approx e^{-\frac{\beta^2}{2} (H(s) - E(s))^2} \nonumber \\&=
\frac{1}{\sqrt{2\pi\beta^2}}\int_{-\infty}^{\infty} e^{-\frac{t^2}{2\beta^2}} e^{-i(H(s)-E(s))t} dt
\label{eq:gaussian_projection}
\end{align}
where $E(s)$ is the energy eigenvalue of \(H(s)\) corresponding to $\ket{\Phi(s)}$. This approach
requires two additional parameters: $\beta$, which determines the
width of the projection operator, and the number of Monte Carlo
samples $N_\nu$ to construct integral in Eq.~\eqref{eq:gaussian_projection}. The accuracy of this approximation and the
associated resources are governed by these two parameters.
Furthermore, the energy eigenvalue \(E(s)\) is determined using the method described in
Sec.~\uppercase{ii} of the SM.

Next, we analyze the computational cost of the proposed algorithm in terms of circuit depth and iteration count.
In this analysis, the target segment length $\Delta l_t$ is treated as a constant parameter based on the following reasoning.
Since the primary role of the constant geometric speed schedule is to distribute the evolution time in proportion to the geometric length, it is sufficient that $\Delta l_t$ remains small compared to the total path length $L$ to obtain a practical advantage.
This reasoning is supported by our numerical experiments (Figs.~\ref{fig:2}--\ref{fig:4}), where a fixed value of $\Delta l_t = 0.2$ was employed and demonstrated consistent scaling improvements.

  \begin{algorithm}
         \refstepcounter{algorithm}\label{alg:css_qzmc}
        \vspace{0.0ex}\hrule height 0.5pt
        \vspace{0.5ex}

        \noindent\textbf{Algorithm \thealgorithm } Segmented CGS Schedule Construction

        \vspace{0.5ex}\hrule height 0.5pt
        \vspace{0.5ex}
        \begin{algorithmic}
          \State \textbf{Input:} Target segment length $\Delta l_{t}$, time for the first
   step $t_1$.
          \State \textbf{Output:} Segmented CGS schedule $\hat{s}_c(\tau)$.
          \State \textbf{Initialize:} $j\gets0$, $t_0\gets0$, $s_0\gets0$, reference
  length $\Delta l_{r} \gets 0$.
          \State \textbf{Initialize:} Schedule points $\mathcal{S} \gets \{(t_0, s_0)\}$.
          \While{$s_{j}<1$}
            \State Construct the schedule $\hat{s}(t)$ by interpolating points in
  $\mathcal{S}$.
            \State Prepare $\ket{\psi(t_j)}$ by evolving $\ket{\Phi(s_0)}$ for time $t_j$
   using $\hat{s}(t)$.
            \State Define $g(s') \gets {p_j(s')}/{f_j} - (1-\Delta
  l_{t}^2)$.
            \State Find $s_{j+1} > s_j$ that solves $g(s_{j+1}) = 0$.
            \State Compute $\Delta l_{j+1} \gets \sqrt{1 -p_j(s_{j+1}) / f_j}$.
            \If{$j=0$}
              \State Set reference length $\Delta l_{r} \gets \Delta l_1$.
              \State Set next time point $t_{j+1} \gets t_1$.
            \Else
              \State Set next time point $t_{j+1} \gets t_j + t_1 \cdot (\Delta l_{j+1} /
   \Delta l_{r})$.
            \EndIf
            \State Add the new point $(t_{j+1}, s_{j+1})$ to $\mathcal{S}$.
            \State $j \gets j+1$.
          \EndWhile
          \State Set total evolution time $T \gets t_j$.
          \State Construct final schedule $\hat{s}(t)$ by interpolating all points in
  $\mathcal{S}$.
          \State \textbf{Return} normalized schedule $\hat{s}_c(\tau) = \hat{s}(T\tau)$.
        \end{algorithmic}
        \vspace{0.5ex}\hrule height 0.5pt
  \end{algorithm}

The \textit{circuit depth} is determined by the total time evolution
length, which consists of the adiabatic evolution time \( T \)
and the additional time length \(T_\mathrm{ov}\) required to implement the
projection operator. The
latter scales as \( \mathcal{O}(\beta) \)~\cite{mchan2025QZMC},
where \( \beta \) is the parameter in the Gaussian approximation of
the projection operator in Eq.~\eqref{eq:gaussian_projection}.
Therefore, the total time evolution length is
\begin{align}
T_\text{tot} = T + \mathcal{O}(\beta),
\end{align}
Here, \( T\) is a adiabatic evolution time with the constant geometric speed schedule.
For our algorithm to be reliable, the adiabatic evolution time $T$
should be large enough to ensure $f_j > 1-\eta$ for all $j$, with $\eta$ taken 
not too large (e.g., $\eta < 0.5$).

To construct the projector that resolves the target eigenstate from
nearby eigenstates, \( \beta \) should satisfy \( \beta >
\Delta_\text{min}^{-1} \), where \( \Delta_\text{min} \) denotes the minimum energy
gap in the adiabatic path, i.e., \( \Delta_\text{min} =
\operatorname{min}_\tau \Delta(\tau)\). In our numerical experiments,
setting \(\beta = 2 \Delta_\text{min}^{-1}\) was sufficient. Therefore,
\begin{align}
  T_\text{tot} = T + \mathcal{O}(\Delta_\text{min}^{-1}).
\end{align}

This shows that the additional cost from the projection operator,
$\mathcal{O}(\Delta_\text{min}^{-1})$, matches the optimal scaling (Eq.~\eqref{eq:T_lowerbound}) of
adiabatic algorithms with respect to the minimum energy gap.
Therefore, the projection does not worsen the overall scaling of the
evolution time. Further rigorous analysis shows that the required $\beta$ also depends on 
parameters such as the target segment length $\Delta l_t$ and the precision 
$\epsilon$ of the estimate $p_j(s_{j+1})/f_j$, which are independent of the Hamiltonian,
so is not discussed here. Such details are provided in Sec.~\uppercase{ii} of the SM.

The \textit{iteration count} of the algorithm is determined by two
factors: (i) the number of schedule segments \( m \), and (ii) the
number of Monte Carlo evaluations required for each segment.

The number of segments is determined by the overlap condition
\begin{align}
|\braket{\Phi(s_j)|\Phi(s_{j+1})}|^2 =  1 - \Delta l_t^2,\quad
\Delta l_{j+1} = \Delta l_t. \label{eq:overlap_condition}
\end{align}
Since the total path length \( L \) is approximately equal to \(
\sum_j \Delta l_j \), the number of segments is
\begin{align}
m \approx \frac{L}{\Delta l_t}.
\end{align}

The primary computational task within each segment is the
root-finding search for the next point $s_{j+1}$ that satisfies
Eq.~\eqref{eq:overlap_condition}. This search is iterative,
requiring multiple evaluations of the overlap. We define
$\overline{N}_{\text{root}}$ as the average number of these 
evaluations needed per segment. In addition to these 
$\overline{N}_{\text{root}}$ steps, a separate evaluation is
required to compute the fidelity term $f_j$. Therefore, the total
number of Monte Carlo evaluations per segment is approximately
$\overline{N}_{\text{root}} + 1$.
Each of these evaluations uses $N_\nu$ Monte Carlo samples, and the
standard deviation for this statistical estimation is bounded by $\sqrt{2}N_\nu^{-1/2}(1-\eta)^{-1}$ (See
Sec.~\uppercase{ii}\,C in SM for derivation). To ensure the error in
the overlap estimation is below a threshold $\epsilon$, it suffices
to set $N_\nu \geq 2\epsilon^{-2} (1-\eta)^{-2}$. Combining these factors, the total number of Monte Carlo samples
required, $N_\text{total} = m \cdot
 (\overline{N}_{\text{root}} + 1) \cdot N_\nu$, is
\begin{align}
N_{\text{total}} \approx (\overline{N}_{\text{root}} + 1)
\frac{L}{\Delta l_t} 
\frac{2}{\epsilon^2} \frac{1}{(1-\eta)^2}.
\end{align}

Having expressed the repetitions in terms of the root-finding steps, we now theoretically analyze the complexity of $N_{\text{root}}$. In our algorithm, determining the increment $\Delta s$ requires solving $|\braket{\Phi(s_j)|\Phi(s_j + \Delta s)}|^2 = 1 - \Delta l_t^2$. Our implementation utilizes a bracketing strategy followed by a root-finding method (e.g., Brent's method): we start with a safe lower bound $\delta s$ and iteratively double the step size ($\Delta s_{\text{trial}} \leftarrow 2 \Delta s_{\text{trial}}$) until the interval $[0.5 \Delta s_{\text{trial}}, \Delta s_{\text{trial}}]$ encloses the solution, after which the root is refined within this bracket. To guarantee that the initial search does not overlook the solution in regions with large geometric speed $v_{\max} \propto \Delta^{-1}$, the lower bound $\delta s$ must scale as $\delta s \lesssim \Delta l_t / v_{\max} \propto \Delta$. Since $\Delta s$ is at most $1$, the maximum number of doubling steps required to cover the dynamic range scales as
\begin{align}
    N_{\text{bracket}} \leq \log_2\left(\frac{1}{\delta s}\right) = \mathcal{O}(\log(\Delta^{-1})).
\end{align}
Upon completion of the bracketing phase, the root is confined to an interval of width $W = 0.5 \Delta s_{\text{trial}}$. Since the true step size $\Delta s$ falls within this interval ($0.5 \Delta s_{\text{trial}} \le \Delta s \le \Delta s_{\text{trial}}$), the interval width $W$ scales proportionally with $\Delta s$. Consequently, for a target relative precision $\epsilon_{\text{rel}}$, the number of refinement steps is given by $\log_2(W / (\epsilon_{\text{rel}} \Delta s)) \approx \log_2(1/\epsilon_{\text{rel}})$, which is independent of the gap. Therefore, the total complexity of the root-finding procedure per segment becomes
\begin{align}
N_{\text{root}} = \mathcal{O}(\log(\Delta^{-1})), \label{eq:N_root}
\end{align}
confirming that the root-finding overhead scales only logarithmically with the inverse gap.

Next, we analyze the gap scaling of the cumulative runtime, $T_\text{run}$, defined as the sum of the total time evolution lengths over all segments and repetitions.
For the $j$-th segment, the evolution time required for a single circuit instance is given by
\begin{align}
  T_{\text{tot}, j} \approx \frac{j}{m} T + \mathcal{O}(\Delta^{-1}). \label{eq:T_tot_j}
\end{align}
This circuit is repeated $(N_{\text{root},j}+1) N_\nu$ times to estimate the overlap.
Since the number of repetitions $(N_{\text{root},j}+1) N_\nu$ scales at most logarithmically with $\Delta^{-1}$ as discussed previously, the cumulative runtime $T_\text{run} = \sum_{j=0}^{m-1} T_{\text{tot}, j} (N_{\text{root},j}+1) N_\nu$ satisfies
\begin{align}
  T_\text{run} &< \mathcal{O}(\log(\Delta^{-1})) \sum_{j=0}^{m-1} T_{\text{tot}, j} \nonumber\\
  &\lesssim \mathcal{O}(\log(\Delta^{-1})) \left(\frac{mT}{2} + m \mathcal{O}(\Delta^{-1})\right). \label{eq:T_run_sum}
\end{align}
Given that the number of segments scales as $m=\mathcal{O}(L)$, we obtain
\begin{align}
  T_\text{run} < \mathcal{O}(L\log(\Delta^{-1})) \left(T+\mathcal{O}(\Delta^{-1})\right). \label{eq:T_run_L}
\end{align}
If $L$ exhibits worst-case scaling ($L \sim \Delta^{-1}$), the adiabatic evolution time $T$ is lower-bounded by $\mathcal{O}(\Delta^{-2})$, leading to a cumulative runtime of $T_\text{run} \sim \mathcal{O}(\Delta^{-3})$ (up to logarithmic corrections). This scaling is indeed less favorable than that of the standard linear schedule.
However, in cases where $L$ remains bounded independently of $\Delta$, $T_\text{run}$ scales linearly with $T$ (up to a logarithmic factor).
Consequently, if the CGS schedule achieves $T=\mathcal{O}(\Delta^{-1})$, the cumulative runtime also follows $T_\text{run} \sim \mathcal{O}(\Delta^{-1})$.
This confirms that the quadratic speedup of our method persists even after considering the total cost.


In addition to the above theoretical analysis, we specify here the
parameter choices used in our practical implementation of the algorithm
for the numerical experiments in Figs.~\ref{fig:2}--\ref{fig:4}.
We used \( \beta = 2 \Delta^{-1} \) and estimated $T_\mathrm{ov} = 2\beta = 4\Delta^{-1}$
as a reasonable overhead for the projection operations,
because at most two projection operations are required to compute each overlap.
Moreover, we used \( \Delta l_t = 0.2\)
and \( N_\nu = 10^4 \), which yielded accurate
schedules with moderate computational cost. 
Under these
parameter settings, the average number of root-finding steps per
segment was observed to be \( \overline{N}_{\text{root}} <20 \).

Finally, we highlight several practical aspects of the algorithm.
First, in many applications, the minimum gap \( \Delta \) is not
known beforehand. For such cases, a suitable value of 
\( \beta \) can be found empirically by starting with a small value
and iteratively increasing it (e.g., along a geometric sequence
\(\beta_k = 2^k \beta_0\)) until convergence of the resulting
schedule is reached. Second,
the initial segment time \( t_1 \) sets the overall timescale for
the evolution via the relation \( T \approx t_1 \cdot (L / \Delta
l_t) \). A suitable value for \( t_1 \) can be estimated from
the application of Eq.~\eqref{eq:infidelity_bound} through a perturbative expansion \(\ket{\Phi(s)}\) near \(s=0\)
or can be adjusted iteratively until the
fidelity \( f_1 \) exceeds a desired threshold. Furthermore, the
root-finding step in Algorithm~\ref{alg:css_qzmc} need not be highly
 precise. Theorem~\ref{theorem:css} does not require the segment
lengths \( \Delta l_j \) to be uniform; it is sufficient to
identify a point \( s_{j+1} \) that yields an overlap close to the
target value, so only a few root-finding steps are necessary. Lastly, for
constructing a continuous schedule from the discrete
points, we employed monotonic cubic spline
interpolation~\cite{Fritsch1984}, as implemented in
SciPy~\cite{SciPy2020}. Our numerical results in
Figs.~\ref{fig:2}-\ref{fig:4} show that while this interpolation
method can introduce small variations in the final fidelity, it does
 not affect the scaling of the adiabatic evolution time
with respect to the minimum gap.

\end{document}


\AtBeginDocument{\pdfcatalog{/OpenAction null}}

\title{Supplemental Material for ``The Constant Geometric Speed Schedule for Adiabatic State Preparation"}

\author{Mancheon Han}
\email{mchan@kias.re.kr}
\affiliation{School of Computational Sciences, Korea Institute for Advanced Study (KIAS), Seoul, 02455, Korea}
\author{Hyowon Park}
\affiliation{Materials Science Division, Argonne National Laboratory, Argonne, IL, 60439, USA}
\affiliation{Department of Physics, University of Illinois at Chicago, Chicago, IL, 60607, USA}
\author{Sangkook Choi}
\email{sangkookchoi@kias.re.kr}
\affiliation{School of Computational Sciences, Korea Institute for Advanced Study (KIAS), Seoul, 02455, Korea}
\date{\today}

\maketitle

\section{Proofs of Equations and Theorems in the main text}
\subsection{Proof of $\|Q\ddot P P\|=\sqrt{\dot v^2+\kappa^2 v^4}$}
Here we provide a proof of $\|Q\ddot P P\|=\sqrt{\dot v^2+\kappa^2 v^4}$
that is used to derive Eq.~\eqref{eq:infidelity_bound_with_speed}.  
Consider the arc-length \(l\) parametrization  
of the path $\{(H(s(\tau)), \ket{\Phi(s(\tau))})\}$,
where \(l\) can be written as
\[
l = \int_0^{\tau} \|\ket{\dot{\Phi}(\tau)}\|\, d\tau.
\]  
We denote by $'$ the derivative with respect to $l$.  
For convenience, we adopt the phase convention $\braket{\Phi(l)|\Phi'(l)} = 0$.  
The curvature of the path is then defined as~\cite{Alsing2024}  
\begin{align}
    \kappa^2(l) = \|Q(l)\ket{\Phi''(l)}\|^2 = \braket{\Phi''(l)|Q(l)|\Phi''(l)},
\end{align}
where $Q(l) = I - \ket{\Phi(l)}\bra{\Phi(l)}$ is the projection operator onto the subspace orthogonal to $\ket{\Phi(l)}$.  
Note that $\|\ket{\Phi'(l)}\|=1$, which implies
\begin{align}
    \braket{\Phi'(l)|\Phi''(l)} + \braket{\Phi''(l)|\Phi'(l)} = 0. \label{eq:orthonormality}
\end{align}

It follows that $\ket{\dot{\Phi}(\tau)}$ can be written as
\begin{align*}
    \ket{\dot{\Phi}} = \frac{d \ket{\Phi}}{d\tau} =
    \frac{dl}{d\tau} \frac{d \ket{\Phi}}{dl} = v(\tau) \ket{\Phi'},
\end{align*}
and similarly
\begin{align*}
    \ket{\ddot{\Phi}(\tau)} &= \frac{d \ket{\dot{\Phi}}}{d\tau} 
    = \frac{d}{d\tau} \left(v(\tau) \ket{\Phi'}\right) \\
    &= \dot{v}(\tau) \ket{\Phi'} + v(\tau) \frac{dl}{d\tau}\ket{\Phi''} \\
    &= \dot{v}(\tau) \ket{\Phi'} + v^2(\tau) \ket{\Phi''}.
\end{align*}

On the other hand,
\[
    \|Q \ddot{P} P\| = \|Q \ket{\ddot{\Phi}}\|.
\]
Using Eqs.~\eqref{eq:orthonormality} and the above, we obtain
\begin{align*}
    \|Q \ket{\ddot{\Phi}}\|^2 
    = \braket{\ddot{\Phi}|Q|\ddot{\Phi}} 
    = \dot{v}^2 \braket{\Phi'|Q|\Phi'} + v^4 \braket{\Phi''|Q|\Phi''}.
\end{align*}
Since $\braket{\Phi'|Q|\Phi'} = 1$ and $\braket{\Phi''|Q|\Phi''} = \kappa^2$, we find
\begin{align*}
    \|Q \ket{\ddot{\Phi}}\|^2 = \dot{v}^2  + \kappa^2 v^4,
\end{align*}
which completes the proof.

\subsection{Proof of Theorem~\ref{theorem:css}}
The proof of theorem~\ref{theorem:css} relies on the following lemma,
\begin{lemma}
\label{lemma:path_segment_overlap}
Let \( \ket{\Phi(s)} \) be a normalized and differentiable family of eigenstates of a parameter-dependent Hamiltonian \( H(s) \), describing an adiabatic path for \( s \in [0,1] \). Then, to the leading order in \( \Delta s \), the overlap between nearby eigenstates satisfies
\begin{align}
  |\braket{\Phi(s) | \Phi(s + \Delta s)}|^2 = 1 - (\Delta l)^2, \label{eq:path_segment_length}
\end{align}
where \( \Delta l\) is defined in Eq.~\eqref{eq:delta_l}.
\end{lemma}

\begin{proof}
By differentiating \( \braket{\Phi(s) | \Phi(s)} = 1 \) with respect to \( s \), we obtain
\[
\partial_s \braket{\Phi(s) | \Phi(s)} = 2 \operatorname{Re} \braket{\Phi(s) | \partial_s \Phi(s)} = 0.
\]
The global phase of \( \ket{\Phi(s)} \) is arbitrary,
so we fix the gauge such that
\(
\braket{\Phi(s) | \partial_s \Phi(s)}
\)
is real. Then the above equation implies
\[
\braket{\Phi(s) | \partial_s \Phi(s)} = 0.
\]
Differentiating this condition once more gives
\[
\braket{\Phi(s) | \partial_s^2 \Phi(s)} = -\left\| \partial_s \ket{\Phi(s)} \right\|^2.
\]

Using a second-order Taylor expansion of \( \ket{\Phi(s + \Delta s)} \), we obtain
\[
|\braket{\Phi(s) | \Phi(s + \Delta s)}|^2 = 1 - \left\| \partial_s \ket{\Phi(s)} \right\|^2 (\Delta s)^2 + \mathcal{O}((\Delta s)^3).
\]

Meanwhile, the path segment length satisfies
\[
\Delta l = \left\| \partial_s \ket{\Phi(s)} \right\| \Delta s + \mathcal{O}((\Delta s)^2).
\]
Substituting this into the overlap expression yields
\[
|\braket{\Phi(s) | \Phi(s + \Delta s)}|^2 = 1 - (\Delta l)^2 + \mathcal{O}((\Delta s)^3),
\]
which proves the lemma.
\end{proof}

Then, we can prove Theorem~\ref{theorem:css}, which is stated again here for clarity.

\begin{restate}
Suppose a monotonic sequence \( s_0 = 0 < s_1 < \dots < s_m = 1 \) is given, and that the eigenstate overlaps 
\[
|\braket{\Phi(s_j) | \Phi(s_{j+1})}|^2 = 1 - (\Delta l_{j+1})^2
\]
are known for each adjacent pair \( (s_j, s_{j+1}) \). Then, any schedule
\( \hat{s}(\tau) \) satisfying \( \hat{s}(\tau_j) = s_j \) and  
\[
\tau_{j+1} = \tau_j + \frac{\Delta l_{j+1}}{\sum_{i=1}^m \Delta l_i}
\]
converges to the constant geometric speed schedule \( s_c \) in the limit \( \max_j (s_{j+1} - s_j) \rightarrow 0 \).
\end{restate}

\begin{proof}
By construction, the average speed on each interval \([s_j,s_{j+1}]\) is given by
\[
\frac{\Delta l_{j+1}}{\tau_{j+1} - \tau_j} = \sum_{i=1}^{m} \Delta l_i.
\]
As \( \Delta s_j \rightarrow 0 \), we have \( \tau_{j+1} - \tau_j \rightarrow 0 \), and the left-hand side converges to the instantaneous speed in normalized time \( \tau \):
\[
\frac{\Delta l_{j+1}}{\tau_{j+1} - \tau_j}
\rightarrow \left\| \frac{d}{d\tau} \ket{\Phi_s(\tau)} \right\|_{\tau=\tau_j}.
\]
At the same time, the right-hand side satisfies \( \sum_{i=1}^m \Delta l_i \rightarrow L \).
Therefore, in the limit \( \max_j (s_{j+1} - s_j) \rightarrow 0 \),
\[
v_s(\tau) = \left\| \ket{\dot{\Phi}_s(\tau)} \right\| = L,
\]
so we have the constant geometric speed schedule.
\end{proof}

\subsection{Constant Geometric Speed Schedule for Grover's Search}
\begin{theorem}
The optimal schedule for the adiabatic Grover problem is a
constant geometric speed schedule, $s_c(\tau)$. \label{theorem:grover}
\end{theorem}
\begin{proof}
  Since $H_i = \mathbb{I} - \ket{\phi}\bra{\phi}$ and $H_f=\mathbb{I}-\ket{m}\bra{m}$,  
the evolution is confined to the subspace spanned by $\ket{\phi}$ and $\ket{m}$.  
A convenient basis for this subspace is $\{\ket{m},\ket{m^\perp}\}$, where  
$\ket{m^\perp}=(N-1)^{-1/2}\sum_{i\neq m} \ket{i}$.  
Here $N=2^n$ denotes the Hilbert space dimension for $n$ qubits.  
In this basis, 
\begin{align*}
  H(s) &= \frac{\mathbb{I}}{2}- \frac{1}{2} \left[ v_z(s) Z + v_x (s) X \right],\\
  v_z(s) &= 1 - 2 (1-s) \left(1-\frac{1}{N}\right),\\
  v_x(s) &= (1-s)\frac{2}{\sqrt{N}} \sqrt{1-\frac{1}{N}},
\end{align*}
where $Z$ and $X$ denote the Pauli matrices.  
The energy gap is $\Delta(s) = \sqrt{v_x^2 + v_z^2}$ and the ground state is  
\(\ket{\Phi(s)} = [\cos(\theta(s)/2), \, \sin(\theta(s)/2)]^T\),  
with $\cos(\theta) = v_z/\Delta$ and $\sin(\theta) = v_x/\Delta$.  

From this expression, constant geometric speed corresponds to $d\theta/dt$ being constant.  
Since $\theta$ is a monotonically decreasing function, we set $d\theta/dt = -c$ with $c>0$.  
Differentiating $\cos(\theta) = v_z/\Delta$ with respect to $t$ gives
\begin{align*}
    -\frac{d\theta}{dt} \sin(\theta) = \frac{ds}{dt} \left(\frac{(\partial_s v_z) \Delta-(\partial_s\Delta) v_z}{\Delta^2}\right),
\end{align*}
where $\partial_s$ denotes differentiation with respect to $s$.  
Substituting $d\theta/dt = -c$ and $\sin(\theta) = v_x/\Delta$, we obtain
\begin{align*}
    c = \frac{ds}{dt} \frac{(\partial_s v_z) v_x - v_z (\partial_s v_x)}{v_z^2+v_x^2}.
\end{align*}
Using the above definitions, this becomes
\begin{align*}
    c\frac{dt}{ds} = 2\frac{\sqrt{N-1}}{N} \frac{1}{(2s-1)^2+4s(1-s)/N}.
\end{align*}
Integrating with the boundary condition $s(t_f)=1$ yields
\begin{align*}
    \tau = \frac{t}{t_f} = \frac{1}{2} + \frac{\tan^{-1}(\sqrt{N-1} (2s-1))}{2\tan^{-1}(\sqrt{N-1})}.
\end{align*}
Inverting this relation gives
\[
    s_c(\tau) = \frac{1}{2} + \frac{1}{2\sqrt{N\shortminus1}} \tan\!\left[\left(2
\tau\shortminus1\right)\tan^{-1}(\sqrt{N\shortminus1})\right], 
\]
which coincides with the optimal schedule~\cite{Roland2002,Albash2018}.
\end{proof}

\section{Overlap estimation and error analysis}
In this section, we describe how the Quantum Zeno Monte Carlo~\cite{mchan2025QZMC} method can be used to 
estimate eigenstate overlap, and provide a detailed analysis of the associated errors.

\subsection{Overlap estimation via the approximate projector \texorpdfstring{$P_H^\beta(E)$}{PHbeta(E)}}
Our goal is to estimate the overlap \(|\braket{\Phi_0|\chi}|^2\) between a given state \(\ket{\chi}\) and the target eigenstate \(\ket{\Phi_0}\).  
In the main text we introduced the Gaussian operator \(P_H^\beta(E)\), which approximates the spectral projection onto the eigenstate with energy \(E\).  
In particular, when \(E\) coincides with the target eigenvalue \(E_0\), the operator \(P_H^\beta(E_0)\) acts as an approximate projector onto \(\ket{\Phi_0}\).  
Here, we discuss how to estimate \(E_0\) using \(P_H^\beta(E)\), and analyze the error of the estimated overlap.

In the present work, we did not use the predictor--corrector method proposed in our previous study~\cite{mchan2025QZMC}.  
Instead, we employed the following method to estimate the target eigenenergy at each step.  
Consider a state \(\ket{\chi}\), which can be expressed as
\begin{align}
  \ket{\chi} = \sum_{k} c_k \ket{E_k},
\end{align}
where \(\ket{E_k}\) is an eigenstate of the Hamiltonian \(H\) with eigenvalue \(E_k\).  
Then, using the approximate projection operator  
\[
  P^\beta_H(E) = \exp\!\left[-\tfrac{1}{2}\beta^2(H - E)^2\right],
\]
we define
\begin{align}
  g(E) &= \|P^\beta_H(E)\ket{\chi}\|^2 = \braket{\chi|P^{\sqrt{2}\beta}_H(E)|\chi}  \nonumber \\
       &= \sum_{k} w_k e^{-\beta^2 (E_k - E)^2},
\end{align}
where \(w_k = |c_k|^2\).  

The eigenvalue \(E_0\), corresponding to the target eigenstate \(\ket{\Phi_0}\),  
can then be estimated by optimizing \(g(E)\).  
Moreover, the estimated energy \(E^*\) converges exponentially to the true eigenvalue \(E_0\) as \(\beta \to \infty\),  
and \(g(E^*)\) provides a reliable estimate of the overlap \(|\braket{\Phi_0|\chi}|^2\),  
which also converges exponentially to the true overlap.  
To prove this, we use an upper-bounding quadratic polynomial of \(g(E)\), as established by the following lemma.  

\begin{lemma}
  \label{lemma:quadratic_upperbound}
  Let \(h(x)\) be a twice differentiable function on \([a,b]\).  
  Define the quadratic function
  \begin{align}
    u_c(x) = h(c) + h'(c)(x-c) + \frac{M}{2}(x-c)^2, \label{eq:quadratic_upperbound}
  \end{align}
  where \(M\) satisfies \(\sup_{x\in[a,b]}|h''(x)| \leq M\) and \(c\in[a,b]\) is arbitrary.
  Then \(h(x) \leq u_c(x)\) for all \(x\in[a,b]\).
\end{lemma}

\begin{proof}
  \begin{align}
    h(x) &= h(c) + \int_c^x h'(y)\,dy  \nonumber \\
         &= h(c) + h'(c)(x-c) + \int_c^x \int_c^y h''(z)\,dz\,dy \nonumber \\
         &\leq h(c) + h'(c)(x-c) + \int_c^x \int_c^y M\,dz\,dy  \nonumber \\
         &= h(c) + h'(c)(x-c) + \frac{M}{2}(x-c)^2. \nonumber
  \end{align}
  Here, $'$ denotes differentiation with respect to \(x\).
\end{proof}

We also require the following lemmas.  

\begin{lemma}
  \label{lemma:quadratic_inequality}
  For \(0 \leq y \leq x^2\), the following inequality holds:
  \begin{align}
    \frac{y}{2x} \leq x - \sqrt{x^2 - y} \leq \frac{y}{x}. \nonumber
  \end{align}
\end{lemma}

\begin{proof}
  Left inequality:
  \begin{align*}
    \left(x-\frac{y}{2x}\right)^2 = x^2 - y + \frac{y^2}{4x^2} \geq x^2 - y.
  \end{align*}
  Right inequality:
  \begin{align*}
    \left(x-\frac{y}{x}\right)^2 = x^2 - 2y + \frac{y^2}{x^2} \leq x^2 - y,
  \end{align*}
  where the last inequality follows from \(x^2 \geq y\), which implies \(y/x^2 \leq 1\) and hence \(y^2/x^2 \leq y\).
\end{proof}

\begin{lemma}
  \label{lemma:gaussian_inequality}
  For \(x \geq 2\), the function
  \[
    f(y) = \frac{x}{x-y}e^{-y^2/2}
  \]
  is increasing in \(y\) whenever \(0 \leq y \leq x/(x^2-1)\).
\end{lemma}

\begin{proof}
  \begin{align*}
    \frac{d}{dy}\!\left(\frac{x}{x-y}e^{-y^2/2}\right) &=
    \frac{x}{x-y}\!\left(\frac{1}{x-y}-y\right) e^{-y^2/2} \\ 
    &= \frac{x}{(x-y)^2} (y^2 -xy + 1) e^{-y^2/2}.
  \end{align*}
  Thus $f(y)$ is increasing when $0 \leq y \leq r_1$,  
  where $r_1$ is the smaller root of $y^2 - xy + 1 = 0$, namely \(2/(x+\sqrt{x^2-4})\).  
  Furthermore,
  \begin{align*} (x^2-2)^2 - x^2 (x^2-4) = 4 >0, \\ x^2-2 > x\sqrt{x^2-4},
    \\ 2x^2-2 > x (x+\sqrt{x^2-4}),
    \\ \frac{2}{x+\sqrt{x^2-4}} > \frac{x}{x^2-1}. \end{align*}
  Therefore, if \(0 \leq y \leq x/(x^2-1)\), the function $f(y)$ is indeed increasing.
\end{proof}
\begin{theorem}
  \label{theorem:energy_estimation_error}
  Let \(\Delta\) denote the spectral gap between the target eigenvalue \(E_0\) and the rest of the spectrum of the Hamiltonian \(H\).  
Define \(r = w_0/\sum_{k\neq0} w_k\).  
Suppose \(\beta\) satisfies
\begin{align}
  \beta > \Delta^{-1} \max\!\left(\sqrt{2},\sqrt{\mathcal{W}(1+c_0 + c_1/r)}\right), \label{eq:beta_condition}
\end{align}
where \(\mathcal{W}(x) = \tfrac{1}{2} - W_{-1}(-\min(0.25\sqrt{e}\,x^{-1},e^{-1}))\),  
\(c_0 = 1.3802\), \(c_1 = 0.3749\), and \(W_{-1}\) denotes the $-1$ branch of the Lambert \(W\) function.  
Then there exists a local maximizer \(E^*\) of \(g(E)\) such that
\begin{align}
  |E^*-E_0| \leq \epsilon_u. \label{eq:energy_bound1}
\end{align}
Here, $\epsilon_u < \Delta/3$ and is determined as
\begin{align}
  \epsilon_u &= \frac{1}{\sqrt{2}\beta}\,
  \frac{a_1 - \sqrt{a_1^2 - 2a_0a_2}}{a_2}, \nonumber \\
  a_0 &= \sqrt{2}\,\beta \Delta e^{-\beta^2\Delta^2}, \nonumber \\
  a_1 &= r - (2\beta^2\Delta^2-1)\,e^{-\beta^2\Delta^2}, \nonumber \\
  a_2 &= c_0 r + c_1. \label{eq:energy_bound2}
\end{align}
Moreover,
\begin{align}
  w_0 \leq g(E^*) \leq w_0 \,\frac{\Delta}{\Delta-\epsilon_u}\, e^{-\beta^2 \epsilon_u^2}. \label{eq:norm_bound}
\end{align}
\end{theorem}
\begin{proof}
  Without loss of generality, set \(E_0=0\).  
  For notational convenience, let \(\sigma = (\sqrt{2}\beta)^{-1}\) denote the standard deviation of each Gaussian term. Then
  \begin{align*}
    g(E) = w_0 e^{-\frac{E^2}{2\sigma^2}} + \sum_{k\neq 0} w_k e^{-\frac{(E_k - E)^2}{2\sigma^2}}.
  \end{align*}
  Its derivative is
  \begin{align*}
    g'(E) = \frac{1}{\sigma^2}\Bigl(-E w_0 e^{-\frac{E^2}{2\sigma^2}}
      + \sum_{k\neq 0} w_k (E_k - E) e^{-\frac{(E_k - E)^2}{2\sigma^2}}\Bigr).
  \end{align*}

  Suppose \(g'(0) > 0\) (the case \(g'(0) < 0\) can be treated analogously).  
  Since \(g'(0) > 0\), a local maximizer \(E^*\) of \(g(E)\) exists with \(E^* \ge 0\).  
  Let \(w_{\ne 0} = \sum_{k\ne 0} w_k\) and define the spectral gap from the target level by
  \(\Delta := \min_{k\ne 0} |E_k|\).
  Consider
  \begin{align*}
    h(E) = w_0 e^{-\frac{E^2}{2\sigma^2}} + w_{\ne 0} e^{-\frac{(\Delta - E)^2}{2\sigma^2}}.
  \end{align*}
  For \(E \in [0,\sigma]\) and every \(k\ne 0\), we have
  \(|E_k - E| \ge |E_k| - |E| \ge \Delta - E\).  
  Hence
  \[
    e^{-(E_k - E)^2/(2\sigma^2)} \le e^{-(\Delta - E)^2/(2\sigma^2)},
  \]
  and summing over \(k\ne 0\) yields \(h(E) \ge g(E)\) on \([0,\sigma]\).

  Moreover, since \(\beta \ge \sqrt{2}/\Delta\), we have \(\sigma = (\sqrt{2}\beta)^{-1} \le \Delta/2\), and thus
  \[
    |E_k - E| \ge |E_k| - |E| \ge \Delta - \sigma \ge \sigma, \qquad E\in[0,\sigma].
  \]
  Therefore, since the function \(x \mapsto x e^{-x^2/(2\sigma^2)}\) is decreasing for \(x \ge \sigma\),
  \begin{align*}
    g'(E)
    &\le \sigma^{-2}\left(-E w_0 e^{-\frac{E^2}{2\sigma^2}}
    + w_{\ne 0} (\Delta - E) e^{-\frac{(\Delta - E)^2}{2\sigma^2}}\right) \\
    &= h'(E),
  \end{align*}
  for all \(E \in [0,\sigma]\).

  Applying Lemma~\ref{lemma:quadratic_upperbound} to \(h'(E)\) with \(c=0\) yields
  \[
    g'(E) \leq h'(E) \leq \frac{u(E)}{\sigma},
  \]
  where
  \begin{align*}
    u(E) &= A_0 - A_1 \frac{E}{\sigma} + \tfrac{1}{2}A_2 \left(\tfrac{E}{\sigma}\right)^2, \\
    A_0 &= \sigma h'(0) = w_{\ne 0} \tfrac{\Delta}{\sigma} e^{-\frac{\Delta^2}{2\sigma^2}}, \\
    A_1 &= -\sigma^2 h''(0) = w_0 - w_{\ne 0}\!\left(\tfrac{\Delta^2}{\sigma^2}-1\right) e^{-\frac{\Delta^2}{2\sigma^2}}, \\
    A_2 &= c_0 w_0 + c_1 w_{\ne 0}.
  \end{align*}
  The constants \(c_0 = 1.3802\) and \(c_1 = 0.3749\) are chosen to satisfy  
  \(\sup_{E\in[0,\sigma]} |h'''(E)| \leq \sigma^{-3}(c_0 w_0 + c_1 w_{\ne 0})\),  
  which follows from a direct calculation of
  \[
    h'''(E) = \sigma^{-3}\bigl(w_0 r(E/\sigma) - w_{\ne 0} r((\Delta-E)/\sigma)\bigr),
  \]
  with \(r(x) = (3x-x^3)e^{-x^2/2}\).  
  Numerical evaluation shows that \(\max_{x\in[0,1]} r(x) \approx 1.3801\) and \(\max_{x\ge 1} -r(x) \approx 0.3748\),  
  consistent with the choice of \(c_0,c_1\).

  Since \(g'(E) \le u(E)/\sigma\), the existence of a root \(\epsilon_u\) of \(u(E)\) in \([0,\sigma]\) guarantees the existence of a root \(\epsilon_h\) of \(h'(E)\) and a root \(E^*\) of \(g'(E)\) with  
  \(0 \leq E^* \leq \epsilon_h \leq \epsilon_u\).  
  To ensure such roots exist, it suffices that \(A_1 > 0\) and \(A_1^2 - 2 A_0 A_2 \ge 0\).  
  These conditions can be reduced to the sufficient condition in Eq.~\eqref{eq:beta_condition} by the following argument.

  First, note that
  \[
    A_0 = w_{\ne 0}\frac{\Delta}{\sigma} e^{-\Delta^2/(2\sigma^2)}
    < w_{\ne 0}\Bigl(\frac{\Delta^2}{\sigma^2}-1\Bigr) e^{-\Delta^2/(2\sigma^2)},
  \]
  because \(\Delta/\sigma > 2\).  
  Denote the right-hand side by \(y\). Then \(A_1 = w_0 - y\).  
  Moreover,
  \[
    0 \le A_1^2 - 2 y A_2 \le A_1^2 - 2 A_0 A_2,
  \]
  so it suffices to require
  \[
    (w_0 - y)^2 - 2 y A_2 \ge 0.
  \]
  This inequality is satisfied if
  \[
    y \le (w_0 + A_2) - \sqrt{(w_0 + A_2)^2 - w_0^2}.
  \]
  By Lemma~\ref{lemma:quadratic_inequality}, this in turn holds if
  \[
    y \le \frac{w_0^2}{2(w_0 + A_2)}.
  \]
  Because \(w_0^2/(w_0 + A_2) \le w_0\), this condition also ensures \(A_1 = w_0 - y \ge 0\).

  Translating this inequality back into a condition on \(\beta\),  
  from the definition of \(y\) we require
  \[
    w_{\ne 0} \Bigl(\frac{\Delta^2}{\sigma^2} - 1\Bigr) e^{-\Delta^2/(2\sigma^2)} \le s,
  \]
  where \(s = w_0^2/(2(w_0 + A_2))\).  
  Multiplying both sides by \(-e^{1/2}/2\) yields
  \[
    w_{\ne 0}\,\frac{1}{2}\Bigl(1 - \frac{\Delta^2}{\sigma^2}\Bigr) e^{\tfrac{1}{2}(1-\Delta^2/\sigma^2)} \ge -\frac{\sqrt{e}}{2}s.
  \]
  This inequality is satisfied if
  \[
    1 - \frac{\Delta^2}{\sigma^2}
      < 2 W_{-1}\!\left(-\min\!\left(\tfrac{\sqrt{e}}{2}s, e^{-1}\right)\right),
  \]
  which is equivalent to
  \[
    \frac{\Delta}{\sigma} > \sqrt{\,1 - 2 W_{-1}\!\left(-\min\!\left(\tfrac{\sqrt{e}}{2}s, e^{-1}\right)\right)}.
  \]
  Since \(\Delta/\sigma = \Delta \sqrt{2}\beta\), this condition is equivalent to
  \[
    \beta > \frac{1}{\sqrt{2}\,\Delta}
      \sqrt{\,1 - 2 W_{-1}\!\left(-\min\!\left(\tfrac{\sqrt{e}}{2}s, e^{-1}\right)\right)}.
  \]
  Substituting \(s = w_0^2/(2(w_0 + A_2)) = 1/(2(c_0+1) + 2c_1/r)\) yields Eq.~\eqref{eq:beta_condition}.  
  Hence the sufficient condition for \(A_1>0\) and \(A_1^2 - 2 A_0 A_2 \ge 0\) is established.

  Under this condition, the smaller root of \(u(E)\) is
  \[
    \epsilon_u = \sigma\,\frac{A_1 - \sqrt{A_1^2 - 2 A_0 A_2}}{A_2},
  \]
  giving Eq.~\eqref{eq:energy_bound1}.  
  Dividing numerator and denominator by \(w_{\ne 0}\) and using \(\sigma^{-1}=\sqrt{2}\beta\) provides the explicit form in Eq.~\eqref{eq:energy_bound2}.

  Furthermore, by Lemma~\ref{lemma:quadratic_inequality},
  \[
    \epsilon_u \le \sigma \frac{2A_0}{A_1}.
  \]
  Since \(w_0^2/(2(w_0+A_2)) \le w_0/3\), we have \(w_0 \ge 3y\), and thus \(A_1 = w_0 - y \ge 2y\).  
  This gives
  \[
    \epsilon_u \le \frac{2A_0}{A_1} \le \frac{2A_0}{2y}
    = \sigma\,\frac{\Delta/\sigma}{\Delta^2/\sigma^2 - 1}.
  \]
  From $\Delta/\sigma>2$, we have $\epsilon_u<\Delta/3$. 

  Moreover, by Lemma~\ref{lemma:gaussian_inequality}, the function
  \(\frac{\Delta}{\Delta-E} e^{-E^2/(2\sigma^2)}\) is increasing for \(E\in[0,\epsilon_u]\).  
  At \(E=\epsilon_h\), where \(h'(E)=0\), we have
  \[
    \epsilon_h w_0 e^{-\epsilon_h^2/(2\sigma^2)}
      = (\Delta-\epsilon_h) w_{\ne 0} e^{-(\Delta-\epsilon_h)^2/(2\sigma^2)}.
  \]
  Substituting this into \(h(\epsilon_h)\) gives
  \[
    h(\epsilon_h) = \frac{\Delta}{\Delta-\epsilon_h} w_0 e^{-\epsilon_h^2/(2\sigma^2)}.
  \]
  Since \(\tfrac{\Delta}{\Delta-E} e^{-E^2/(2\sigma^2)}\) is increasing for \(E\in[0,\epsilon_u]\), we obtain
  \[
    h(\epsilon_h) \le \frac{\Delta}{\Delta-\epsilon_u} w_0 e^{-\epsilon_u^2/(2\sigma^2)}.
  \]
  Finally, because \(g(E^*) \le h(E^*) \le h(\epsilon_h)\), we have
  \[
    g(E^*) \le \frac{\Delta}{\Delta-\epsilon_u} w_0 e^{-\epsilon_u^2/(2\sigma^2)}.
  \]
  On the other hand, since \(w_0 \le g(0)\le g(E^*)\), we also have the lower bound.  
  This establishes Eq.~\eqref{eq:norm_bound}.
\end{proof}
Theorem~\ref{theorem:energy_estimation_error} yields the following corollary.
\begin{corollary}
  \label{corollary:beta_energy}
  Retaining terms up to first order in \(e^{-\beta^2\Delta^2}\), \(\epsilon_u\) in Eq.~\eqref{eq:energy_bound2} becomes
  \begin{align}
    \epsilon_u = r^{-1}\, \Delta e^{-\beta^2\Delta^2}.
  \end{align}
  Therefore, if \(\beta\) satisfies Eq.~\eqref{eq:beta_condition} and
  \begin{align}
    \beta > \Theta\!\bigl(\Delta^{-1}\sqrt{\log(r^{-1}\epsilon^{-1}\Delta)}\bigr), \label{eq:beta_energy}
  \end{align}
  then the estimated energy \(E^{*}\) obeys \(|E^{*}-E_0| < \epsilon\).
\end{corollary}

\begin{proof}
  From Eq.~\eqref{eq:energy_bound2}, \(\epsilon_u\) can be expressed as
  \begin{align*}
    \epsilon_u = \sigma a_1 \frac{1 - \sqrt{1 - 2 a_0 a_2/a_1^2}}{a_2}.
  \end{align*}
  Using \(\sqrt{1-x} = 1 - \tfrac{1}{2}x + \mathcal{O}(x^2)\), we obtain
  \begin{align*}
    \epsilon_u = \sigma \frac{a_0}{a_1}.
  \end{align*}
  up to first order in \(e^{-\beta^2\Delta^2}\). Moreover, since \(a_1 = r + \mathcal{O}(e^{-\beta^2\Delta^2})\),
  \begin{align*}
    \epsilon_u = \sigma \frac{a_0}{r} = r^{-1} \Delta e^{-\beta^2\Delta^2}.
  \end{align*}
  This proves the corollary.
\end{proof}

We also obtain the following corollary for estimating \(w_0\).
\begin{corollary}
  \label{corollary:norm_error}
  If \(\beta\) satisfies Eq.~\eqref{eq:beta_condition} and
  \begin{align}
    \beta > \Theta\!\bigl(\Delta^{-1}\sqrt{\log(r^{-1}(1+\epsilon^{-1}))}\bigr), \label{eq:beta_norm}
  \end{align}
  then there exists a local maximizer \(E^*\) of \(g(E)\) such that \(g(E^*)\) provides an estimate of \(w_0\) with relative error less than \(\epsilon\).
\end{corollary}

\begin{proof}
  From Theorem~\ref{theorem:energy_estimation_error},
  \begin{align*}
    0 \leq g(E^*)-w_0 \leq w_0 \frac{\Delta}{\Delta-\epsilon_u} e^{-\beta^2 \epsilon_u^2} - w_0
    \leq w_0 \frac{\epsilon_u}{\Delta-\epsilon_u}.
  \end{align*}
  Therefore, if \(\epsilon_u \leq \epsilon \Delta/(1+\epsilon)\), then
  \begin{align*}
    0 \leq g(E^*)-w_0 \leq w_0 \epsilon.
  \end{align*}
  From Corollary~\ref{corollary:beta_energy}, this condition is guaranteed by choosing \(\beta\) satisfying Eq.~\eqref{eq:beta_condition} and
  \begin{align*}
    \beta = \Theta\!\bigl(\Delta^{-1}\sqrt{\log(r^{-1}(1+\epsilon^{-1}))}\bigr).
  \end{align*}
  This proves the corollary.
\end{proof}

\subsection{Application to Algorithm~\ref{alg:css_qzmc}}
We now apply Theorem~\ref{theorem:energy_estimation_error} and its corollaries to the case of Algorithm~\ref{alg:css_qzmc} in the main text.  
For an adiabatically evolved state \(\ket{\psi(t_j)}\), we estimate
\[
  f_j = \|P(s_j)\ket{\psi(t_j)}\|^2
\]
using the Gaussian approximation
\begin{align*}
  f_j^\beta &= \bigl\|e^{-\tfrac{\beta^2}{2}(H_j-\mathcal{E}_{j})^2}\ket{\psi(t_j)}\bigr\|^2 \\
  &= \braket{\psi(t_j)|e^{-\beta^2(H_j-\mathcal{E}_{j})^2}|\psi(t_j)},
\end{align*}
where \(H_j = H(s_j)\) and \(\mathcal{E}_j\) is the estimated eigenenergy obtained from the procedure described in the previous section.  

Similarly, we estimate
\[
  p_j(s_{j+1}) = \|P(s_{j+1}) P(s_j)\ket{\psi(t_j)}\|^2
\]
by
\begin{align}
  p_j^\beta(s_{j+1}) &= \braket{\xi_j|e^{-\beta^2(H_{j+1}-\mathcal{E}_{j+1})^2}|\xi_j}, \nonumber \\
  \ket{\xi_j} &= e^{-\tfrac{\beta^2}{2}(H_j-\mathcal{E}_{j})^2}\ket{\psi(t_j)}. 
\end{align}
Theorem~\ref{theorem:energy_estimation_error} and Corollary~\ref{corollary:norm_error} can be applied directly, as in the case of \(f_j^\beta\).  
The main quantity of interest is the parameter \(w_0\) in Theorem~\ref{theorem:energy_estimation_error}, which here takes the form
\[
  w_0 = \braket{\xi_j|\Phi(s_{j+1})}\braket{\Phi(s_{j+1})|\xi_j}.
\]
We denote this quantity by \(\tilde{p}_j\) and derive its upper and lower bounds.

By construction, we have \(|\braket{\Phi(s_j)|\Phi(s_{j+1})}|^2 = 1-\Delta l_t^2\).  
By choosing a phase convention such that \(\braket{\Phi(s_j)|\Phi(s_{j+1})} = \sqrt{1-\Delta l_t^2}\),
we decompose \(\ket{\xi_j}\) as follows:
\begin{align*}
  \ket{\xi_j} &= a \ket{\Phi(s_j)} + b \ket{\Phi^\perp(s_j)}, \\
  f_j^\beta &= |a|^2 + |b|^2 ,
\end{align*}
where \(\ket{\Phi^\perp(s_j)}\) is normalized and orthogonal to \(\ket{\Phi(s_j)}\).  

From Theorem~\ref{theorem:energy_estimation_error}, the energy estimate obeys \(|\mathcal{E}_j| \leq \epsilon_{uf}\), where \(\epsilon_{uf}\) is given by Eq.~\eqref{eq:energy_bound2} with \(r = f_j/(1-f_j)\).  
Consequently,
\begin{align*}
  |a|^2 &\geq f_j e^{-\beta^2 \epsilon_{uf}^2}, \\
  \frac{|a|^2}{f_j^\beta} &\geq \frac{\Delta_j-\epsilon_{uf}}{\Delta_j}, \\
  \frac{|b|^2}{f_j^\beta} &\leq \frac{\epsilon_{uf}}{\Delta_j},
\end{align*}
where \(\Delta_j\) is the minimum energy gap between \(\ket{\Phi(s_j)}\) and the other eigenstates of \(H(s_j)\).  
For convenience, define \(\delta_f = \epsilon_{uf}/\Delta_j\). Then
\[
  \frac{|a|^2}{f_j^\beta} \ge 1-\delta_f, \qquad \frac{|b|^2}{f_j^\beta} \le \delta_f,
\]

We now evaluate
\begin{align*}
  \tilde{p}_j &= |\braket{\Phi(s_{j+1})|\xi_j}|^2 \\
  &= |a \braket{\Phi(s_{j+1})|\Phi(s_j)} + b \braket{\Phi(s_{j+1})|\Phi^\perp(s_j)}|^2.
\end{align*}
Since \(\braket{\Phi(s_{j+1})|\Phi(s_j)} = \braket{\Phi(s_{j})|\Phi(s_{j+1})} = \sqrt{1-\Delta l_t^2}\) and \(|\braket{\Phi(s_{j+1})|\Phi^\perp(s_j)}|\leq \Delta l_t\), we obtain
\begin{align}
  \tilde{p}_j &\geq f_j^\beta\Bigl(\sqrt{1-\Delta l_t^2}\sqrt{1-\delta^{\phantom{2}}_f} - \Delta l_t \sqrt{\delta^{\phantom{2}}_f}\Bigr)^2, \nonumber \\
  \tilde{p}_j &\leq f_j^\beta\Bigl(\sqrt{1-\Delta l_t^2}\sqrt{1-\delta^{\phantom{2}}_f} + \Delta l_t \sqrt{\delta^{\phantom{2}}_f}\Bigr)^2. \label{eq:pjtilde_bound1}
\end{align}
Since $0\leq\Delta l_t<1$, and $0\leq\delta_f<1$ (from Theorem~\ref{theorem:energy_estimation_error}),
we introduce $\theta_1, \theta_2 \in [0,\pi/2)$ defined by
$\sin{\theta_1} = \Delta l_t$, $\sin{\theta_2} = \sqrt{\delta_f}$.
Eq.~\eqref{eq:pjtilde_bound1} can thus be rewritten as
\begin{align}
  \tilde{p}_j/f_j^\beta\geq \cos^2(\theta_1+\theta_2), \nonumber \\
  \tilde{p}_j/f_j^\beta\leq\cos^2(\theta_1-\theta_2).
\end{align}
Using the identity $1-\Delta l^2_t = \cos^2(\theta_1)$ and $\cos^2(\theta_1\pm\theta_2) - \cos^2(\theta_1) = \mp \sin(2\theta_1\pm\theta_2)\sin(\theta_2)$, we have
\begin{align}
\tilde{p}_j/f_j^\beta-(1-\Delta l^2_t)&\geq-\sin(2\theta_1+\theta_2)\sin(\theta_2), \nonumber \\
\tilde{p}_j/f_j^\beta-(1-\Delta l^2_t)&\leq \sin(2\theta_1-\theta_2)\sin(\theta_2).
\end{align}
Let $\theta_2\leq \theta_1$. Then, with $\sin(\theta)<\theta$ for $\theta>0$, we have
\begin{align}
|\tilde{p}_j/f_j^\beta-(1-\Delta l^2_t)|< 3\theta_1 \theta_2. \label{eq:pjtilde_bound2}
\end{align}

Applying Theorem~\ref{theorem:energy_estimation_error} to \(p_j^\beta \equiv p_j^\beta(s_{j+1})\), we find
\[
  \frac{\tilde{p}_j}{f_j^\beta} \;\leq\; \frac{p_j^\beta}{f_j^\beta} \;\leq\; \frac{\tilde{p}_j(1+\delta_p)}{f_j^\beta},
\]
where \(\delta_p = \epsilon_{up}/(\Delta_{j+1}-\epsilon_{up})\), with \(\epsilon_{up}\) the bound from Eq.~\eqref{eq:energy_bound2} evaluated at \(r = \tilde{p}_j/(f_j^\beta-\tilde{p}_j)\).  

Combining this with Eq.~\eqref{eq:pjtilde_bound2}, we obtain
\begin{align*}
  \frac{p_j^\beta}{f_j^\beta} - (1-\Delta l_t^2) &\geq -3 \theta_1 \theta_2, \\
  \frac{p_j^\beta}{f_j^\beta} - (1-\Delta l_t^2) &\leq 3 \theta_1 \theta_2 (1+\delta_p) + \delta_p (1-\Delta l^2_t).
\end{align*}
Thus, we can ensure \(|p_j^\beta/f_j^\beta-(1-\Delta l_t^2)| \leq \epsilon\) from
$3 \theta_1 \theta_2 \leq \epsilon/2, \theta_2\leq\theta_1$, and $(1-\Delta l_t^2) \delta_p \leq \epsilon/3$,
which is satisfied by letting
\begin{align*}
  \delta_f &\leq \min\left(\sin^2\left(\frac{\epsilon}{6\arcsin(\Delta l_t)}\right),\Delta l_t^2\right), \\
  \delta_p &\leq \frac{\epsilon}{3(1-\Delta l_t^2)}.
\end{align*}
From Corollary~\ref{corollary:beta_energy}, the first inequality is satisfied if
\begin{align*}
  \beta > \Theta\!\Bigl(\Delta_j^{-1}\sqrt{\log\gamma_1(f_j,\Delta l_t,\epsilon)}\Bigr), \\
  \gamma_1(f_j,\Delta l_t,\epsilon) = \max\left(\Delta l^{-2}_t, \csc^2\left(\frac{\epsilon}{6\arcsin(\Delta l_t)}\right)\right)\frac{1-f_j}{f_j}.
\end{align*}
If the above condition is satisfied, from Corollary~\ref{corollary:norm_error} and Eq.~\eqref{eq:pjtilde_bound2}, the second inequality is satisfied if
\begin{align*}
  \beta > \Theta\!\Bigl(\Delta_{j+1}^{-1}\sqrt{\log\gamma_2(\Delta l_t,\epsilon)}\Bigr), \\
  \gamma_2(\Delta l_t,\epsilon) = \bigl(1+3(1-\Delta l_t^2)\epsilon^{-1}\bigr)\,\frac{\Delta l_t^2+0.5\epsilon}{1-\Delta l_t^2-0.5\epsilon}.
\end{align*}

Therefore, defining \(\gamma(f_j,\Delta l_t,\epsilon) = \sqrt{\log(\max(\gamma_1,\gamma_2))}\), we conclude that the quantity \(1-\Delta l_t^2\) can be estimated within an error \(\epsilon\) in Algorithm~\ref{alg:css_qzmc}, provided that
\begin{align}
  \beta > \Theta\!\bigl(\Delta^{-1}\,\gamma(1-\eta,\Delta l_t,\epsilon)\bigr), \label{eq:beta_bound_alg}
\end{align}
where \(\eta\in[0,1)\) is such that \(f_j \ge 1-\eta\) for all \(j\).

\subsection{Monte Carlo error for the eigenstate overlap}
Finally, we discuss the Monte Carlo error in the estimation of \(p_j^\beta/f_j^\beta\).  
By the delta method, the variance \(\sigma^2\) of the ratio is
\begin{align}
  \sigma^2 
  = \frac{1}{\bigl(f_j^\beta\bigr)^2}\Bigl(\sigma_{p}^2 + \frac{\bigl(p_j^\beta\bigr)^2}{\bigl(f_j^\beta\bigr)^2}\,\sigma_{f}^2\Bigr),
\end{align}
where \(\sigma_p^2\) and \(\sigma_f^2\) denote the variances of the estimators of \(p_j^\beta\) and \(f_j^\beta\), respectively.  
Since \(\sigma_p^2,\sigma_f^2 \leq 1/N_\nu\) and \(p_j^\beta/f_j^\beta \leq 1\), it follows that
\begin{align}
  \sigma^2 \leq \frac{2}{\bigl(f_j^\beta\bigr)^2 N_\nu}, 
  \qquad 
  \sigma \leq \frac{\sqrt{2}}{f_j^\beta\sqrt{N_\nu}},
\end{align}
where \(N_\nu\) is the number of Monte Carlo samples used in each evaluation of \(f_j^\beta\) and \(p_j^\beta\).  

Therefore, if \(f_j \geq 1-\eta\) for all \(j\), we obtain
\begin{align}
  \sigma \leq \frac{\sqrt{2}}{(1-\eta)\sqrt{N_\nu}}.
\end{align}

\section{Additional data for the numerical experiments}
Here, we provide additional data supplementing the numerical experiments in the main text.
First, Fig.~\ref{fig:S1} presents the calculated schedules corresponding to Fig.~\ref{fig:3}(b) and Fig.~\ref{fig:4}(c).

\begin{figure} 
\centering
\centerline{\includegraphics[width=8.6cm]{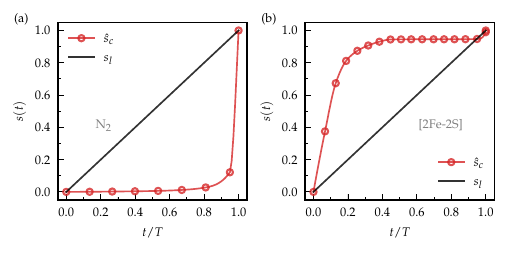}}
\caption{Calculated segmented CGS profiles for
(a) the N$_2$ molecule at a bond length of $R=3.5\,\text{\AA}$ and
(b) the [2Fe-2S] cluster starting from the DFT ground state.}
\label{fig:S1}
\end{figure}
Second, in Fig.~\ref{fig:S2}, we compare the performance of our method with gap-based local optimal schedules
constructed from $ds/d\tau \propto \Delta^p$~\cite{Jansen2007, Roland2002, Albash2018} ($1<p\leq 2$).
The figure demonstrates that our method yields performance similar to these schedules.
\newpage

\begin{figure} 
\centering
\centerline{\includegraphics[width=8.6cm]{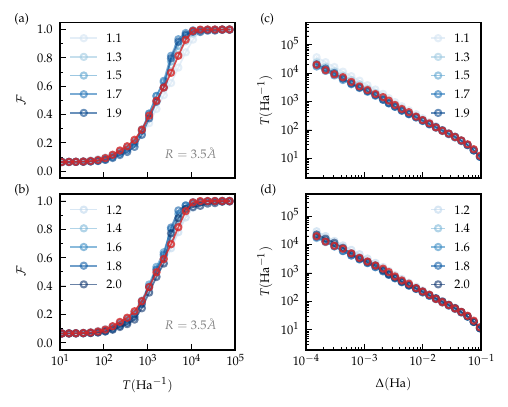}}
\caption{Comparison between the CGS schedule (red) and $\dot{s} \propto \Delta^p$ schedules for the N$_2$ molecule.
(a,b) Fidelity \(\mathcal{F}\) as a function of the evolution time \(T\) at bond length \(R = 3.5\,\text{\AA}\).
(c,d) Evolution time required to reach $\mathcal{F}=0.75$ versus the minimum gap $\Delta$.}
\label{fig:S2}
\end{figure}